\documentclass[aip,jmp]{revtex4-1}

\usepackage[utf8,latin1]{inputenc}
\usepackage[T1]{fontenc}     
\usepackage[american,ngerman,greek,brazilian,british]{babel}
\usepackage{float}
\usepackage[dvipsnames]{xcolor}
\usepackage[colorlinks=true,citecolor=Green,linkcolor=Red,urlcolor=Cyan,hyperindex]{hyperref}
\usepackage{cleveref}
\usepackage{graphicx}
\usepackage{epsfig}
\usepackage{color}
\usepackage{lmodern}
\usepackage{mathpazo}
\usepackage[babel=true]{microtype}
\usepackage{amsmath,amssymb,amsthm}
\usepackage{mathtools}
\usepackage{bm}
\usepackage{enumerate}
\usepackage{sidecap}
\usepackage{enumitem}
\usepackage{listings}
\usepackage{ulem}

\newcount\colveccount
\newcommand*\colvec[1]{
        \global\colveccount#1
        \begin{pmatrix}
        \colvecnext
}
\def\colvecnext#1{
        #1
        \global\advance\colveccount-1
        \ifnum\colveccount>0
                \\
                \expandafter\colvecnext
        \else
                \end{pmatrix}
        \fi
}

\newcommand{\tr}{\text{\normalfont Tr}}

\newcommand{\ie}{\textit{i.e.}}
\newcommand{\eg}{\textit{e.g.}}

\newtheorem{lemma}{Lemma}
\newtheorem{theo}{Theorem}
\newtheorem{prop}{Proposition}
\theoremstyle{definition}
\newtheorem{example}{Example}
\newtheorem{coro}{Corollary}

\newcommand{\cH}{\mathcal{H}}
\newcommand{\Pos}{\text{Pos}}

\newcommand{\bN}{\textbf{N}}
\newcommand{\bM}{\textbf{M}}

\newcommand{\II}{\mathbb{I}}

\newcommand{\cB}{\mathcal{B}}
\newcommand{\RR}{\mathbb{R}}

\newcommand{\cA}{\mathcal{A}}
\newcommand{\bA}{\textbf{A}}
\newcommand{\bB}{\textbf{B}}

\newcommand{\beq}{\begin{equation}}
\newcommand{\eeq}{\end{equation}}

\newcommand{\ufmgfis}{Departamento de F\'isica, Universidade Federal de Minas Gerais, 31270-901, Belo Horizonte, MG, Brazil}
\newcommand{\ufmgmat}{Departamento de Matem\'atica, Universidade Federal de Minas Gerais, 31270-901, Belo Horizonte, MG, Brazil}
\newcommand{\icfo}{ICFO-Institut de Ciencies Fotoniques, Mediterranean Technology Park, 08860 Castelldefels, Barcelona, Spain}
\newcommand{\unicamp}{Departamento de Matem\'atica Aplicada, IMECC-Unicamp, 13084-970, Campinas, S\~ao Paulo, Brazil}
\newcommand{\iqoqi}{Institute for Quantum Optics and Quantum Information (IQOQI), Austrian Academy of Sciences, Boltzmanngasse 3, A-1090 Vienna, Austria}
\newcommand{\icrea}{ICREA-Instituci\'o Catalana de Recerca i Estudis Avan\c cats, Lluis Companys 23, 08010 Barcelona, Spain}

\begin{document}

\title{Operational framework for quantum measurement simulability} 

\author{Leonardo Guerini}
\email[]{leonardo.guerini@icfo.eu}
\affiliation{\ufmgmat}
\affiliation{\icfo}

\author{Jessica Bavaresco}
\affiliation{\iqoqi}
\affiliation{\ufmgfis}

\author{Marcelo Terra Cunha}
\affiliation{\unicamp}

\author{Antonio Ac\'\i n}
\affiliation{\icfo}
\affiliation{\icrea}


\begin{abstract}
We introduce a framework for simulating quantum measurements based on classical processing of a set of accessible measurements.
Well-known concepts such as joint measurability and projective simulability naturally emerge as particular cases of our framework, but our study also leads to novel results and questions.
First, a generalisation of joint measurability is derived, which yields a hierarchy for the incompatibility of sets of measurements. 
A similar hierarchy is defined based on the number of outcomes necessary to perform a simulation of a given measurement.
This general approach also allows us to identify connections between different kinds of simulability and, in particular, we characterise the qubit measurements that are projective-simulable in terms of joint measurability.
Finally, we discuss how our framework can be interpreted in the context of resource theories.
\end{abstract}


\maketitle

\section{Introduction}\label{intro}

In the last decades much of the research on quantum theory has focused on quantum states, exploring topics such as entanglement theory and state estimation protocols.
Of no less importance, quantum measurements also present a rich collection of properties that still remain to be fully understood, many of them essential to reveal remarkable features of the theory (\eg \ Bell nonlocality~\cite{bell64,brunner14}, uncertainty relations~\cite{heisenberg27, wehner10}), or to achieve optimality in certain tasks (\eg \ quantum state tomography~\cite{renes04}, quantum state discrimination~\cite{wootters89, bergou10}).

Quantum measurements are modelled by positive-operator-valued measures (POVMs).
Two of the most important properties a set of POVMs may present are compatibility and projectiveness.
Measurement compatibility, or joint measurability, is a property that generalises the notion of commutativity. 
A set of measurements is jointly measurable whenever their statistics can be reproduced by post-processing the statistics obtained by a single POVM. 
This property ensures that a set of measurements leads to results that can be classically modelled in certain scenarios~\cite{kruszyski87,teiko08}. 
Joint measurability is intimately connected to EPR steering~\cite{quintino14,uola14}, an important class of quantum correlations that can be used to certify entanglement in a semi-device independent manner~\cite{wiseman07}.

Projective measurements are related to quantum observables and historically played a major role in quantum theory.
Although projective POVMs are favourable for experimental implementations and present a simpler mathematical structure, they are outperformed by non-projective measurements in various tasks~\cite{renes04, bergou10, bennett92, vertesi10, derka98, acin16}.
Naimark's dilation theorem~\cite{naimark40} guarantees that any non-projective POVM can be implemented as a projective measurement on the original system plus an ancilla.
However, the simulation of general POVMs by projective ones in the same dimension, that is, without the help of any ancilla, is a topic that has only recently been considered~\cite{oszmaniec16}.
By fixing the dimension, this characterises those POVMs that could demonstrate advantages over projective ones.
This is the definition of projective simulability we follow here.

In this work we consider the question of measurement simulability: given a set of accessible measurements, which other measurements can be simulated by them when assisted by classical pre- and post-processing? 
We provide a general operational framework to study this problem and see how joint measurability and projective simulability appear as particular cases of it.
Within our study of different forms of measurement simulability, we provide useful tools for identifying how measurements can be more efficiently implemented on an experimental setup, by optimising the number of POVMs to be performed, their number of outcomes, or by restricting them to be projective.

Our framework also provides a common background to understand relations between types of simulability.
From the equivalence of simulability by a single POVM and compatibility, we introduce a generalisation of the concept of joint measurability by increasing the number of simulators, which defines different degrees on the incompatibility of sets of measurements.
We also show that if a given set of measurements shares the same set of simulators and pre-processing step, then they must be jointly measurable. 
This motivates us to study the simulation of POVMs via measurements with less outcomes and we characterise this type of simulability in terms of joint measurability.
This leads to a perhaps surprising connection between joint measurability and projective simulability:
we prove that a qubit POVM $\bA$ can be simulated by projective measurements if and only if it can be jointly measured with its Bloch-antipodal POVM $\bar{\bA}$ (\textit{i.e.} the POVM whose effects $\bar{A}_i$ have Bloch vectors antipodal to the ones of $A_i$).

Finally, we interpret our framework from a resource-theoretical point of view~\cite{vedral97, plenio07, brandao13, ahmadi13, gour08, devicente14}, where non-simulability plays the role of resource.
In this approach, we show that classical processing and white noise robustness are suitable choices of free operations and resource measure, respectively.

\section{Preliminaries}

We start by introducing our notation and mathematical framework.
Let $\cH$ be a finite dimensional Hilbert space and $\Pos(\cH)$ the set of positive semidefinite operators acting on $\cH$.
A \textit{quantum measurement} on $\cH$  corresponds to a positive-operator valued measure (POVM), that is, a tuple $\bA = (A_1, \ldots, A_n)\in \Pos(\cH)^{\times n}$ of positive semidefinite operators satisfying $\sum_a{A_a} = \II$, where each $A_a$ corresponds to outcome $a$, $n$ is the number of outcomes, and $\II$ is the identity operator on $\cH$.
The operators $A_a$ are called the \textit{effects} of $\bA$.
In the case where the effects $A_a$ are projectors, we say that $\bA$ is a \textit{projective measurement}.
Notice that some effects might be null, corresponding to outcomes that never occur.

A measurement $\bA$ can be simulated by a subset of POVMs $\cB = \{\bB^{(j)}\}_j$ if there is a protocol based on classical manipulations of the measurements in $\cB$ that yields the same statistics as $\bA$ when performed on any quantum state,
\beq
\text{Pr}_{\text{prot}}(i|\rho) = \tr(A_i\rho),
\eeq
for any outcome $i$ and any state $\rho$.

Quantum measurements can be classically manipulated in two ways~\cite{haapasalo2012}: as a \textit{pre-processing} (mixing) and as a \textit{post-processing} (relabeling).
Here we restrict ourselves to operations only on the level of the measurements, although pre-processing operations involving the preparation of quantum states could also be defined~\cite{buscemi05}.
Therefore, the most general protocol for simulating $\bA$ with $\cB$ consists in three steps:
\begin{enumerate}
\item[(i)] Choose a measurement $\bB^{(j)}\in\cB$ with probability $p(j|\bA)$;
\item[(ii)] Perform $\bB^{(j)}$;
\item[(iii)] Upon obtaining outcome $i'$, output $i$ according to some probability $q(i|\bA,j,i')$.
\end{enumerate}

In the above protocol, step (i) represents a pre-processing and step (iii) represents a post-processing.
In the latter, the final output $i$ is produced with a probability $q(i|\bA,j,i')$ conditioned on the POVM $\bA$ to be simulated, on the performed measurement $\bB^{(j)}$, and on the obtained outcome $i'$.
This can be understood as a new measurement $\tilde\bB^{(j)}$ given by effects \cite{stupidfootnote}
\beq\label{football}
\tilde B^{(j)}_{i} = \sum_{i'}q(i|\bA,j,i')B^{(j)}_{i'}.
\eeq
Notice that $\tilde \bB^{(j)}$ may have a different number of outcomes than $\bB^{(j)}$ (either more or less).

Step (i) allows for probabilistic mixing of the post-processed POVMs $\tilde \bB^{(j)}$.
Therefore, we say that an $n$-outcome POVM $\bA$ is \textit{$\cB$-simulable} if there are probability distributions $p(\cdot|\bA), q(\cdot|\bA,j,i')$ such that for any state $\rho$,
\begin{eqnarray}\nonumber
\tr(A_i\rho) &=& \text{Pr}_{\text{prot}}(i|\rho)\\
&=& \tr\left(\left[\sum_j{p(j|\bA)}\sum_{i'}{q(i|\bA,j,i')B^{(j)}_{i'}}\right]\rho\right),
\end{eqnarray}
or, equivalently,
\beq\label{princessbubblegum}
A_i = \sum_j{p(j|\bA)}\sum_{i'}{q(i|\bA,j,i')B^{(j)}_{i'}},
\eeq
for $i\in\{1,\ldots,n\}$.
In this case, we say that the particular choice of measurements $\bB^{(j)}$ involved in the above decomposition are the \textit{$\cB$-simulators} of $\bA$.

It is straightforward to see that any trivial POVM $\bA = (a_1\II,\ldots,a_n\II)$ that has effects proportional to the identity can be simulated only with classical post-processing, simply by taking $q(i|j,i')=a_i$, for all $j,i'$, and therefore are simulable by any set $\cB$ of simulators.
This leads us to the study of the robustness of a given POVM $\bA$.

By applying the \textit{depolarising map}
\beq
\Phi_t: A \mapsto tA + (1-t)\tr(A)\frac\II d
\eeq
to each effect of $\bA$, for some $t\in [0,1]$, we obtain a depolarised version of the measurement,
\beq
\Phi_t(\bA) := (\Phi_t(A_1),\ldots,\Phi_t(A_n)).
\eeq
The parameter $t$ is called the \textit{visibility} of $\bA$ in $\Phi_t(\bA)$.
The depolarising map can be physically interpreted as the presence of white noise in the implementation of $\bA$, and therefore its consideration is natural from an experimental point of view.
We will focus on white noise, but other models of noise could be considered and even optimised, such as in the case of generalised robustness\cite{piani15}.

Notice that the completely depolarised version of $\bA$,
\beq
\Phi_0(\bA) = (\tr(A_1)\II/d,\ldots,\tr(A_n)\II/d),
\eeq
is a trivial POVM, and therefore simulable by any set of measurements.
Then we can define the white noise robustness of $\bA$ regarding its simulation by $\cB$ as
\beq
t_\cB^\bA = \max\{t;\ \Phi_t(\bA)\ \text{is $\cB$-simulable}\}.
\eeq

After introducing all the previous concepts, we are now in position to present our results.

\section{Limiting the number of simulators}\label{lumpyspacekingdom}

The main goal of the next sections is to study the $\cB$ simulability of general POVMs, under different sets of simulators, depending on the number or type of measurements, or the number of outcomes.
We start by considering completely general accessible measurements, restricting solely the number of simulators.

\subsection{Simulability by a single measurement}

Perhaps the simplest form of simulation refers to the case where the subset $\cB$ of measurements to which one has access contains a single POVM $\bB$ of $n_B$ outcomes.
In this case, step (i) of the general protocol is trivial, and the only relevant operation is post-processing.
Therefore, the $\cB$-simulable POVMs are the ones described in Eq. (\ref{football}).

When we consider a set of $m$ measurements $\{\bA^{(l)}\}_{l=1}^m$ that are simulable by the same (arbitrary) POVM $\bB$ we recover the usual definition of joint measurability, as already pointed in Ref.~\cite{ali09}.
Indeed, consider that 
\beq
A_i^{(l)} = \sum_{i'}{q(i|l,i')B_{i'}},
\eeq
for all $i,l$ and some post-processings $q(\cdot|l,i'),\ i'\in\{1,\ldots,n_B\}$.
Then define a \textit{joint measurement} $\bM$ by
\beq\label{scorchy}
M_{a_1\ldots a_m} = \sum_{i=1}^{n_B}\prod_{l=1}^{m}{q(a_l|l,i)B_{i}},
\eeq
for $a_l\in\{1,\ldots,n_l\}$, where $n_l$ is the number of outcomes of $\bA^{(l)}$.
Hence
\begin{equation}\label{gunther}
\sum_{r\neq l}\sum_{a_r=1}^{n_r} M_{a_1\ldots(a_l=i)\ldots a_m} = A^{(l)}_i,
\end{equation}
for all $l, i$, and we obtain the usual definition of joint measurability: the set $\{\bA^{(j)}\}$ is \textit{jointly measurable}, or \textit{compatible}, since all POVM elements $A^{(j)}_i$ can be recovered by (deterministically) coarse-graining over the joint measurement $\bM$.
This proves the following lemma.
\begin{lemma}\label{cosmicowl}
A set of POVMs is jointly measurable if and only if it can be simulated by a single measurement.
\end{lemma}

Joint measurability thus appears as a particular instance of measurement simulability where only one simulator is considered.
The joint measurement $\bM$ derived from $\bB$ simplifies the post-processing at the cost of typically increasing the number of outcomes of the simulator.


If we can simulate a set of POVMs using only one POVM we will say that the set is \textit{single-POVM-simulable}, as an easily generalisable synonymous of jointly measurable. 
By depolarising each POVM in a set of POVMs $\cA=\{\bA^{(j)}\}$ we can define its depolarised version,
\beq
\Phi_t(\cA) := \{\Phi_t(\bA^{(j)})\},
\eeq
and its white noise robustness regarding single-POVM simulability (or joint measurability),
\beq\label{bananaguard}
t_{\text{1-POVM}}^\cA = \max\{t;\ \Phi_t(\cA)\ \text{is single-POVM-simulable}\}.
\eeq

One can efficiently decide on the single-POVM simulability of a given set of measurements via a semidefinite program (SDP), an efficiently solvable class of optimisation problems.
In fact, since the only requirements for the joint measurement are positive semidefinitiveness and the linear contraints in Eq. (\ref{gunther}), this problem can be phrased as a feasibility SDP~\cite{wolf09}.
A simple modification of it can be used to calculate the white noise robustness of such set~\cite{cavalcanti16}.

\subsection{Simulability by many measurements}

The natural next step is now to consider a set of simulators containing two POVMs, $\cB= \{\bB^{(1)},\bB^{(2)}\}$. 
Again we look at sets of POVMs $\cA = \{\bA^{(l)}\}$ that can be simulated by the same simulators, \ie, for every effect $A_i^{(l)}$ we have
\beq\label{simon}
A_i^{(l)} = p(1|l)\sum_{i'}q(i|l,1,i')B^{(1)}_{i'}+p(2|l)\sum_{i'}q(i|l,2,i')B^{(2)}_{i'}.
\eeq
Following Eq. (\ref{scorchy}), using deterministic post-processing, this is equivalent to
\begin{eqnarray}\nonumber
A_i^{(l)} &=& p(1|l)\sum_{r\neq l}\sum_{a_r=1}^{n_r} M^{(1)}_{a_1\ldots(a_l=i)\ldots a_m}\\ &&+ p(2|l)\sum_{r\neq l}\sum_{a_r=1}^{n_r} M^{(2)}_{a_1\ldots(a_l=i)\ldots a_m}.
\end{eqnarray}
Hence, in terms of joint measurability, now we can combine the marginals of two joint measurements $\bM^{(1)},\bM^{(2)}$.

In contrast with the previous case, we were unable to cast the problem of deciding whether a given set of measurements is 2-POVM-simulable as an SDP.
Since the variables are the pre-processing, the simulators, and the post-processing, Eq. (\ref{simon}) represents apparently unavoidable non-linear constraints.

Since single-POVM simulability is equivalent to joint measurability, by increasing the number of simulators in the accessible set $\cB$ we create a hierarchy of simulability protocols where each case strictly contains the previous one and whose first level is joint measurability.
However, we now show that if the POVMs $\{\bA^{(j)}\}$ can be simulated by $\cB$ using always the same weights $p(1|l),p(2|l)$ in Eq. (\ref{simon}) independent of $l$, then this set is jointly measurable, and therefore simulable by a single POVM.
This is a general feature of the framework, valid for any set of simulators $\cB$.

\begin{prop}\label{prismo}
If every measurement in $\{\bA^{(j)}\}$ is $\cB$-simulable with the same pre-processing step, then $\{\bA^{(j)}\}$ is jointly measurable.
\end{prop}
\begin{proof}
The proof is analogous to the one of Lemma \ref{cosmicowl}.
If a $\cB$-simulable set shares the same pre-processing, then $p(j|l)=p(j)$ and we can describe its elements by
\beq
A_i^{(l)} = \sum_j p(j) \sum_{i'}{q(i|l,j,i')B^{(j)}_{i'}}.
\eeq
Hence a joint measurement $\bM$ is defined by
\beq
M_{a_1\ldots a_m} = \sum_j p(j) \sum_{i=1}^{n_{j}}\prod_{l=1}^{m}{q(a_l|l,j,i)B_{i}^{(j)}}.
\eeq
\end{proof}
Similarly to Lemma \ref{cosmicowl}, under the conditions of Proposition \ref{prismo} we can exchange many simulators by a single one, generally with a greater number of outcomes.  
In Section \ref{breakfastkingdom} we apply Proposition \ref{prismo} to more specific cases.

Considering simulability with more than one simulator we can refine our notion of incompatibility, as illustrated by the following example.

\begin{example}
Consider the set $\cA=\{\bA^{(x)}, \bA^{(y)}, \bA^{(z)}, \bA^{(\Sigma)}\}$, where $\bA^{(x)},\bA^{(y)},\bA^{(z)}$ are the projective qubit measurements associated to the Pauli observables $\sigma_x,\sigma_y,\sigma_z$ and $\bA^{(\Sigma)}$ is the projective measurement described by $A^{(\Sigma)}_{\pm} = (\II \pm \vec v \cdot \vec \sigma)/2$, with $\vec v  = (1,1,1)/\sqrt 3$.
Now, our goal is to understand for which values of the visibility $t$ the set $\Phi_t(\cA)$ becomes single-, 2- and 3-POVM-simulable. 
Let us start by the latter.

For 3-POVM simulability, a straightforward protocol can be obtained for visibilities in which a pair of POVMs of $\Phi_t(\cA)$ becomes jointly measurable. 
This happens at $t_{\text{PI}} = 0.7420$, where $\bA^{(\Sigma)}$ becomes jointly measurable with any of the other three measurements in the set. 
For visibilities larger than $t_{\text{PI}}$ the set is pairwise incompatible, as there is no pair of POVMs in $\Phi_t(\cA)$ which is jointly measurable. 
However, we next show that this protocol is not optimal for 3-POVM simulability.

Since one of the three-element subsets of $\cA$ is clearly more incompatible than the others (namely, $\{\bA^{(x)},\bA^{(y)},\bA^{(z)}\}$), a better strategy to simulate $\cA$ with 3 simulators is to assign each element of this subset to an exclusive simulator.
This means that for these measurements each pre-processing is deterministic,
\begin{eqnarray}
p(j|\bA^{(w)}) = \delta_{j,1}\delta_{x,w} + \delta_{j,2}\delta_{y,w}+\delta_{j,3}\delta_{z,w},
\end{eqnarray}
where $w=x,y,z$, and each $\bA^{(w)}$ is simulated by a single simulator $\bB^{(j)}$, while $\bA^{(\Sigma)}$ uniformly combines all three simulators,
\beq
p(j|\bA^{(\Sigma)}) = \frac 1 3,\ j=1,2,3.
\eeq

By fixing this pre-processing, we can now write an SDP to calculate the best post-processing steps corresponding to it and the best parameter $t$ such that $\Phi_t(\cA)$ is simulated by this protocol.
With this strategy, we find that the set is 3-POVM-simulable at visibility $t_{\text{3-POVM}} = 0.7746$.
Note that for this value of the visibility we have constructed a particular simulation protocol employing three measurements.
It is in principle conceivable that a better simulation protocol exists, which would imply a larger range for 3-POVM simulation.
Yet this protocol was enough to show a gap with the value required to observe pairwise joint measurability.

At visibility $t_{\text{2-POVM}} = 1/\sqrt 2\approx 0.7071$, $\Phi_{t_{\text{2-POVM}}}(\cA)$ becomes 2-POVM-simulable.
This coincides with the visibility $t_{PC}$ needed to make $\cA$ pairwise compatible, identifying it as a ``hollow tetrahedron'', that is, a set of four incompatible POVMs from which every pair of elements is compatible.
Indeed, since any pair of POVMs of $\cA$ is compatible, we can use the joint measurements $\bM^{(xy)}$ (for depolarised versions of $\bA^{(x)}$ and $\bA^{(y)}$), and $\bM^{(z\Sigma)}$ (for depolarised versions of $\bA^{(z)}$ and $\bA^{(\Sigma)}$) as simulators, each one simulating its corresponding pair.

$\cA$ is triplewise incompatible for visibilities $t\geq t_{TI}=0.6236$.
The set becomes triplewise compatible at visibility $t_{\text{TC}}=1/\sqrt 3 \approx 0.5774$, and, finally, fully compatible when depolarised by a parameter of $t_{\text{1-POVM}}=0.5730$.
Recall that these values $t_{TI}, t_{TC},$ and $t_{\text{1-POVM}}$ are obtained via SDP.
 
A brute force numerical search supports the claim that $t_{\text{2-POVM}}, t_{\text{3-POVM}}$ are the optimal parameters for 2- and 3-POVM simulability of $\cA$, respectively.
On Figure \ref{jamesbaxter} we organise all optimal visibilities for the simulability of $\cA$.

\begin{figure}[h!]
\begin{center}
\includegraphics[width=350pt]{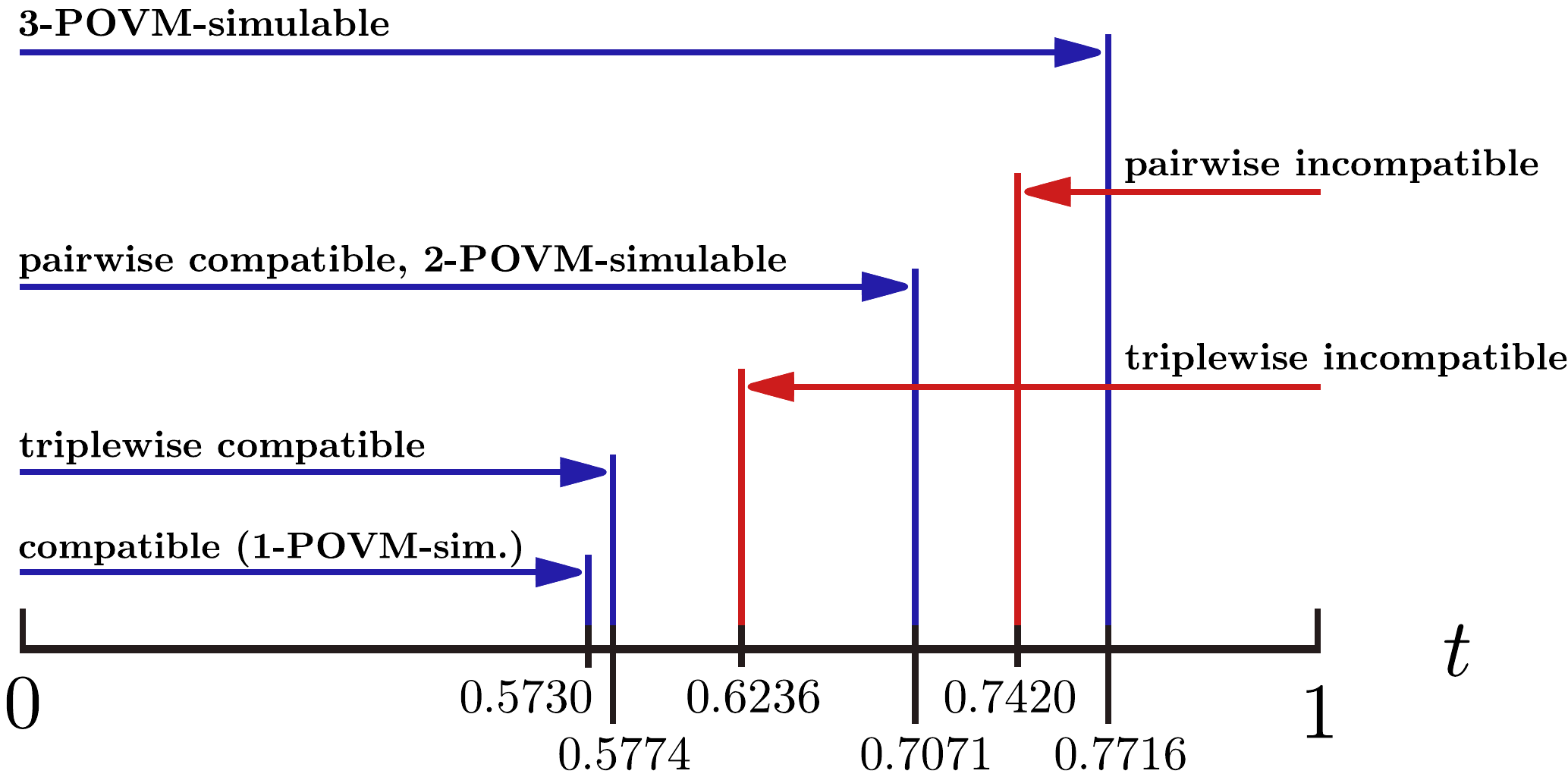}
	\caption{\small{The optimal visibilities for the single-, 2-, and 3-POVM simulability of $\cA=\{\bA^{(x)},\bA^{(y)},\bA^{(z)}, \bA^{(\Sigma)}\}$. Note that the intervals where the set is, say, pairwise compatible and pairwise incompatible are not complementary because these concepts address every possible pair of the set, and different pairs present different degrees of robustness.}}\label{jamesbaxter}
\end{center}
\end{figure}
\end{example}

On the one hand, the above example shows that the number of simulators available yields genuinely different forms of simulability.
On the other, it makes clear that internal compatibility relations between the POVMs of the set provide lower bounds for its $J$-POVM simulability.
For instance, for any set $\{\bA^{(l)}\}$ of $m$ measurements we have
\beq
t_{\text{PI}} \leq t_{\text{(m-1)-POVM}},
\eeq
where $t_{\text{PI}}$ defines open interval of visibilities for which the set is pairwise incompatible, and $t_{\text{(m-1)-POVM}}$ is the critical depolarising parameter for which the set becomes $(m-1)$-POVM-simulable.
Indeed, at $t=t_{PI}$ some pair of POVMs is compatible, say $\bA^{(1)}$ and $\bA^{(2)}$, and thus we can use the simulators $\bB^{(1)}=\bM^{(12)}$ (the joint measurement for $\Phi_t(\bA^{(1)}), \Phi_t(\bA^{(2)})$), $\bB^{(2)}=\bA^{(3)},\ldots,\bB^{(m-1)}=\bA^{(m)}$.
Similarly, we can derive other bounds related to $t_{TI}, t_{PC}$ and so on.

More generally, for a set $\{\bA^{(l)}\}$ of $m$ incompatible POVMs we can consider its robustness regarding simulability with any number $J<m$ of simulators.
For the particular case $J=1$, the noise robustness of joint measurability was extensively studied already~\cite{heinosaari15, cavalcanti16, bavaresco17}, but for $J>1$ this is a new question to be investigated.

\section{Limiting the number of outcomes of the simulators}\label{firekingdom}

Another form of simulability we investigate is by POVMs of less outcomes.
In this case, we do not limit the number of accessible measurements employed for the simulation, but only their number of outcomes.
In other words, now our set of simulators $\cB$ is the set of $k$-outcome POVMs on dimension $d$, and the $\cB$-simulable measurements will be called $k$-outcome-simulable.
This topic arises naturally as another variant of the general simulation problem, and this sort of limitation plays a key role in Bell nonlocality scenarios~\cite{brunner14, kleinmann16, masanes05}.

Note that by applying Lemma \ref{cosmicowl} and Proposition \ref{prismo} one reduces the number of simulators but raises the number of outcomes of the simulators; now we want to improve on the other direction and reduce the number of outcomes, possibly by increasing the number of involved measurements.

We show now that in $k$-outcome simulability the post-processing step can be assumed to be quite simple.
Consider a protocol in which we perform measurement $\bB$ and upon obtaining outcome $i'$, we output $i_0$ with probability $p$ and $i_1$ with probability $(1-p)$.
This is equivalent to the protocol in which with probability $p$ we perform $\bB^{(0)}=\bB$ and always relabel outcome $i'$ by $i_0$, and with probability $(1-p)$ we perform $\bB^{(1)}=\bB$ and always relabel $i'$ by $i_1$.
By doing this for each outcome $i'$, we artificially increase the number of simulators but restrict the post-processing to be deterministic; since we want to simulate an $n$-outcome POVM via $k$-outcome POVMs ($k<n$), we are left with $\frac{n!}{(n-k)!}$ possibilities of post-processing.
Now we notice that post-processing operations that shuffle the order of effects can also be mapped to the pre-processing, in the sense that for each of the $k!$ post-processing permutations on the non-null outcomes we associate a different simulator with permuted effects, which is also an $k$-outcome POVM.
Hence we do not lose generality by considering only the $\frac{n!}{k!(n-k)!}={n \choose k}$ deterministic post-processing strategies that takes $k$-outcome POVMs to $n$-outcome ones while preserving the relative order of effects.

Finally, we can group the simulators that share the same post-processing.
Indeed, imagine that $\bA$ is simulated by the $k$-outcome POVMs $\{\bB^{(j)}=(B^{(j)}_1, \ldots,B^{(j)}_k)\}$ and $\bB^{(1)},\bB^{(2)}$ after being post-processed have the form $\tilde\bB^{(j)}= (B^{(j)}_1, \ldots,B^{(j)}_k,0,\ldots,0),j=1,2$.
Then
\begin{eqnarray}
\bA = p(1)\tilde\bB^{(1)}+p(2)\tilde\bB^{(2)}+\sum_{j>3}{p(j)\tilde\bB^{(j)}} = (p(1)+p(2))\tilde\bB' +\sum_{j>3}{p(j)\bB^{(j)}},
\end{eqnarray}
where the measurement $\tilde\bB'$ is given by 
\beq
\tilde B_i=\frac{p(1)B^{(1)}_i+p(2)B^{(2)}_i}{p(1)+p(2)}
\eeq
also has the form $\tilde \bB = (\tilde B_1, \ldots,\tilde B_k,0,\ldots,0)$.
We conclude that we can consider only one representant of each post-processing class and arrive at the following result, already presented in Ref.~\cite{oszmaniec16}.
\begin{lemma} \label{beemo}
An $n$-outcome POVM $\bA$ is $k$-outcome-simulable if and only if there is a set of at most ${n \choose k}$ POVMs $\{\bB^{(j)}\}$ with at most $k$ non-null effects satisfying
\beq
\bA = \sum_{j=1}^{n \choose k}p_j \bB^{(j)},
\eeq
one for each possible distribution of the $k$ non-null outcomes among the $n$ possibilities.
\end{lemma}

This Lemma allows one to efficiently decide on the $k$-outcome simulability of a given POVM and to compute the amount of depolarisation the POVM endures before becoming $k$-outcome-simulable by means of SDP~\cite{oszmaniec16}.

\begin{example}
Consider a tetrahedral qubit measurement $\bA^{\text{tetra}}$ given by $A^{\text{tetra}}_i = (\II + \vec v_i \cdot \vec \sigma)/4, \ i\in\{1,\ldots,4\}$, where the unit vectors $\vec v_i \in \RR^3$ form the vertices of a regular tetrahedron.
This 4-outcome POVM is not 3-outcome-simulable, but when depolarised by $t^{\text{tetra}}_\text{3-out} = 2\sqrt 2/3$, we see that the resulting POVM, $\Phi_{t^{\text{tetra}}_\text{3-out}}(\bA^{\text{tetra}})$, can be decomposed into ${4 \choose 3}=4$ trine POVMs, $\bB^{\text{trine},r}, \ r\in\{1,\ldots,4\}$, each one with effects whose Bloch vectors point in the direction of a equilateral triangle on the plane perpendicular to $\vec v_r$ (Figure \ref{jamesII}).
\end{example}

\begin{figure}[h!]
\begin{center}
	\includegraphics[width=350pt]{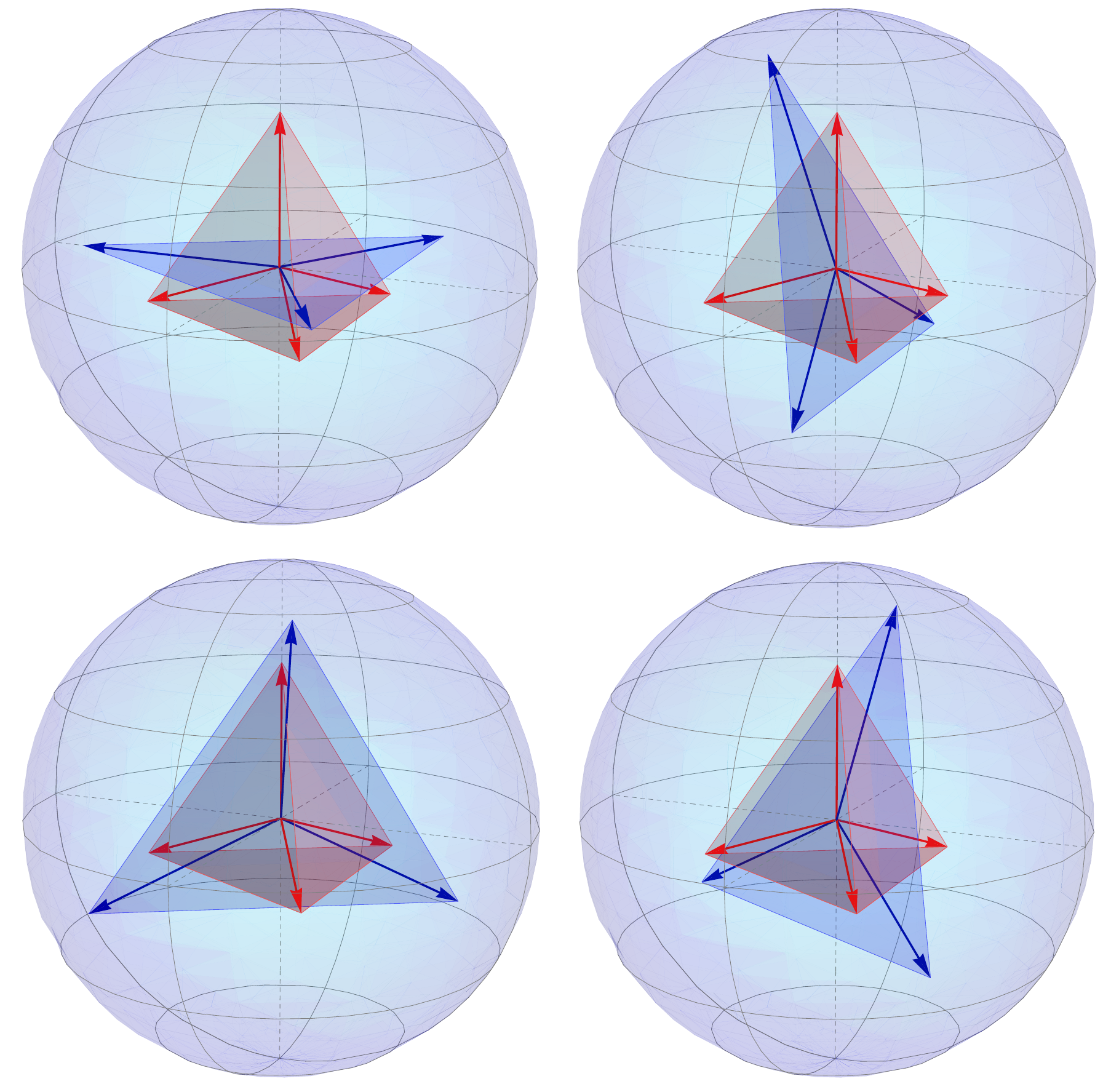}
	\caption{\small{The 4-outcome measurement $\bA^{\text{tetra}}$ becomes 3-outcome-simulable when depolarised by a parameter $t^{\text{tetra}}_\text{3-out} = 2\sqrt 2/3$. Its optimal 3-outcome simulators are regular trines measurements, each one lying on a plane parallel to a facet of the tetrahedron.}}\label{jamesII}
\end{center}
\end{figure}

The trine POVMs $\bB^{\text{trine},r}$ are not 2-outcome-simulable, but this can be achieved by depolarising them by $t^{\text{trine}}_\text{2-out} = \sqrt 3/2$.
It is known that the critical visibility to make $\bA^{\text{tetra}}$ 2-outcome-simulable is $t^{\text{tetra}}_\text{2-out}=\sqrt{2/3}$~\cite{oszmaniec16, hirsch17}, and therefore in this case we have
\beq
t^{\text{tetra}}_\text{3-out}\cdot t^{\text{trine}}_\text{2-out} = t^{\text{tetra}}_\text{2-out}.
\eeq
However, in general one value is only a lower bound for the other,
\begin{equation}
t^{\bA}_\text{$k$-out}\cdot \min\{t^{\bB}_\text{($k-1$)-out} ; \bB\ \text{is a $k$-outcome simulator of $\bA$}\}  \leq t^\bA_\text{$(k-1)$-out}.
\end{equation}

\subsection{$k$-outcome simulability and joint measurability}\label{adventuretime}

For the joint measurability of two POVMs $\{\bA^{(1)},\bA^{(2)}\}$, there is a very visual way of interpreting the joint POVM $\bM$.
If we organise the effects of $\bM$ in an $n\times n$ table, where the effect $M_{a_1a_2}$ occupies position $(a_1,a_2)$, then the marginals correspond to summing over the rows and columns, and Eq. (\ref{gunther}) can be represented by
\begin{eqnarray}\label{orgalorg}
\begin{tabular}{ccc|c}
$M_{11}$ & $\cdots$ & $M_{1n}$ & $A^{(1)}_1$\\
$\vdots$ & $\ddots$ & $\vdots$ & $\vdots$\\ 
$M_{n1}$ & $\cdots$ & $M_{nn}$ & $A^{(1)}_n$\\  
\hline 
$A_1^{(2)}$ & $\cdots$ & $A_n^{(2)}$ &
\end{tabular}.
\end{eqnarray}

Our next result shows that reorganising the effects of simulators in tables representing joint measurements leads to an equivalent condition to $k$-outcome simulability in terms of joint measurability. 

\begin{prop}\label{peppermintbutler} A qudit measurement $\bA$ is $k$-outcome-simulable if and only if there is a joint POVM $\bM$ for the pair $\{\bA, \vec p \cdot \II\}$ with at least $n-k$ null effects in each column $(M_{1j},\ldots,M_{nj})$, where $\vec p \cdot \II = (p_1 \II,\ldots, p_{n \choose k} \II)$.
\end{prop}
\begin{proof}
Suppose we have a decomposition of a POVM $\bA$ into $k$-outcome POVMs $\bB^{(j)}$,
\beq\label{leech}
p_1\colvec{3}{B^{(1)}_1}{\vdots}{B^{(1)}_n} + \cdots +
p_m\colvec{3}{B^{(m)}_1}{\vdots}{B^{(m)}_n} = \colvec{3}{A_1}{\vdots}{A_n},
\eeq
where at least $n-k$ effects are null in each $\bB^{(j)}$, and $m = {n \choose k}$ according to Proposition \ref{beemo}.
We can now group the weights $p_j$ with each effect $B^{(j)}_{b_j}$ and organise these effects in a table
\begin{eqnarray}
\begin{tabular}{ccc|c}
$p_1B^{(1)}_1$ & $\cdots$ & $p_mB^{(m)}_1$ & $A_1$\\
$\vdots$ & $\ddots$ & $\vdots$ & $\vdots$\\ 
$p_1B^{(1)}_n$ & $\cdots$ & $p_mB^{(m)}_n$ & $A_n$\\  
\hline
\end{tabular}.
\end{eqnarray}
Due to the normalisation of the $\bB^{(j)}$, summing over each column we obtain $p_j\II$, and analogously to Table (\ref{orgalorg}), we can see the table as a joint POVM for $\bA$ and $\vec p \cdot \II$.

On the other hand, every joint POVM for $\bA$ and $\vec p \cdot \II$ with $n-k$ null effects in each column can generate a decomposition like Eq. (\ref{leech}), where each column represents one of the $k$-outcome simulators.
\end{proof}

Though any POVM is jointly measurable with a trivial POVM $\vec p \cdot \II$ having all effects proportional to the identity, Proposition \ref{peppermintbutler} is a criterion that requires this compatibility to be given in an optimised way where the joint measurement has many null effects, in order to ensure $k$-outcome simulability.

\subsection{The antipodal measurement}\label{breakfastkingdom}

One of the main advantages of studying different forms of simulation on the same framework is that it facilitates the comprehension of connections between them.
In this subsection we illustrate this by presenting relations between $k$-outcome simulability and joint measurability (single-POVM simulability).
Our starting point is to check the consequences of Proposition \ref{prismo} for $k$-outcome simulability.

Consider the simple case of an $n$-outcome POVM $\bA$ which is 2-outcome-simulable.
Then, according to Lemma \ref{beemo}, there are ${n \choose 2}$ convex weights $(p_{ij})$ and dichotomic POVMs $\bB^{(ij)}=(B_{ij},\II - B_{ij})$, $(i,j) \in \{(1,2),\ldots,(n-1,n)\}$, that can be embedded in the set of $n$-outcome POVMs via post-processing, such that $B_{ij}$ takes place on the $i$-th entry of the tuple and $\II-B_{ij}$ on the $j$-th.
Thus we can write
\begin{align}\label{susanstrong}
\colvec{5}{A_1}{A_2}{A_3}{\vdots}{A_n} =
p_{12}\colvec{5}{B_{12}}{\II-B_{12}}{0}{\vdots}{0}+
p_{13}\colvec{5}{B_{13}}{0}{\II-B_{13}}{\vdots}{0}+
\ldots + p_{(n-1)n}\colvec{5}{0}{\vdots}{0}{B_{(n-1)n}}{\II-B_{(n-1)n}}
\end{align}
Notice that this is equivalent to write each effect of $\bA$ as 
\begin{equation}\label{jakethedog}
A_i = \sum_{j;j<i}{p_{ji}(\II-B_{ji})} + \sum_{j;j>i}{p_{ij}B_{ij}},
\end{equation}
and that each effect $A_i$ is the sum of only ${{n-1} \choose {2-1}}=n-1$ non-null operators $p_{ij}B_{ij}$ or $p_{ji}(\II -B_{ji})$.

According to Proposition \ref{prismo}, if we maintain the same pre-processing $(p_{ij})$ on the right-hand side of Eq. (\ref{susanstrong}) but change the post-processing that embeds the dichotomic measurements, the resulting POVM will be jointly measurable with $\bA$.
Now consider the post-processing of $\bB^{(ij)}$ that takes $B_{ij}$ to the $j$-th position and $\II-B_{ij}$ to the $i$-th position.
This way, we construct another 2-outcome-simulable POVM $\tilde\bA$ given by


\begin{equation}\label{finnthehuman}
\tilde{A}_i = \sum_{j;j<i}{p_{ji}B_{ji}} + \sum_{j;j>i}{p_{ij}(\II-B_{ij})},
\end{equation}
in contrast with Eq. (\ref{jakethedog}).
Proposition $\ref{prismo}$ says that
\begin{equation}\label{treetrunks}
M_{a_1a_2} = 
\begin{cases}
p_{a_1a_2}B_{a_1a_2},\ \ &\text{if}\ \ a_1<a_2 \\
0,\ \ &\text{if}\ \ a_1=a_2  \\
p_{a_1a_2}(\II-B_{a_1a_2}),\ &\text{if}\ \ a_1>a_2
\end{cases}.
\end{equation} 
defines a joint measurement for $\{\bA, \tilde \bA\}$. For example, for $n=4$ this joint measurement reads
\begin{eqnarray}\label{iceking}
\begin{tabular}{cccc|cc}
0 & $p_{12}B_{12}$ & $p_{13}B_{13}$ & $p_{14}B_{14}$ & $A_1$ &\\ 
$p_{12}(\II-B_{12})$ & 0 & $p_{23}B_{23}$ & $p_{24}B_{24}$ & $A_2$ &\\ 
$p_{13}(\II-B_{13})$ & $p_{23}(\II-B_{23})$ & 0 & $p_{34}B_{34}$ & $A_3$&\\ 
$p_{14}(\II-B_{14})$ & $p_{24}(\II-B_{24})$ & $p_{34}(\II-B_{34})$ & 0 & $A_4$ &\\ 
\hline 
$\tilde A_1$ & $\tilde A_2$ & $\tilde A_3$ & $\tilde A_4$ & &
\end{tabular}.
\end{eqnarray}

A drawback in the definition of $\tilde \bA$ is that we cannot construct it directly from $\bA$, since it depends on the simulators $\bB^{(ij)}$ and the pre-processing $(p_{ij})$.
We can avoid this by restricting more the simulation and imposing that the simulators $\bB_{ij}$ are \textit{unbiased} 2-outcome POVMs\cite{busch09}, meaning that each effect has the same parcel of identity when decomposed into a Hermitian operator basis,
\begin{subequations}\label{lumpspaceprincess}
\beq
B_{ij} = \frac12\II + \vec{v}_{ij}\cdot \vec{\lambda}
\eeq
\beq
\label{lumpspaceprincess2} \II - B_{ij} = \frac12\II - \vec{v}_{ij}\cdot \vec{\lambda}
\eeq
\end{subequations}
where $\vec{v}_{ij}\in \RR^{d^2-1}$.
Here, $\vec{\lambda}$ is a vector of $d^2-1$ Hermitian traceless operators that, together with $\II$, form a basis for the real vector space of Hermitian operators in dimension $d$~\cite{bertlmann08} (\eg\ the Pauli matrices for $d=2$, and the Gell-Mann matrices for $d=3$).
We call $\vec{\lambda}$ the \textit{generalised Pauli vector}.

If that is the case and $A_i = a_i\II + \vec{u}_i\cdot\vec\lambda$, then Eq. (\ref{jakethedog}) yields
\begin{eqnarray}
& a_i & = \sum_{j;j<i}{\frac{p_{ji}}{2}} + \sum_{j;j>i}{\frac{p_{ij}}{2}},\\
& \vec{u}_i & = \sum_{j;j<i}{p_{ji}(-\vec{v}_{ji})} + \sum_{j;j>i}{p_{ij}\vec{v}_{ij}},
\end{eqnarray}
and from Eq. (\ref{finnthehuman}) we have that $\tilde{A}_i = a_i\II - \vec{u}_i\cdot\vec\lambda$.
In other words, $\tilde{\bA}$ can be defined directly from $\bA$ by flipping the sign of the generalised Pauli vector of each effect, when the latter is simulable via unbiased dichotomic POVMs.
This motivates the definition of antipodal operator: given an Hermitian operator $A=a\II + \vec{v}\cdot\vec\lambda$, its \textit{antipodal operator} is $\bar{A}=a\II - \vec{v}\cdot\vec\lambda$. 

Since the antipodal POVM $\bar{\bA}$ can be constructed from the simulators of $\bA$ (Eq. (\ref{finnthehuman})), the proof of Proposition \ref{prismo} ensures that $\bar{A}_i \geq 0$.
However, the antipodal of a positive semidefinite operator is not always positive semidefinite, this will generally depend on the eigenvalues of the traceless operator $v\cdot\vec\lambda$. 
An exception is the qubit case, where $d=2$; in this case $\vec\lambda = \vec\sigma$ is the usual vector of Pauli matrices and it holds that $a\II + \vec{v}\cdot\vec\lambda \geq 0$ if and only if $a\geq||\vec{v}||$, which implies that $A=a\II + \vec{v}\cdot\vec\lambda \geq 0$ if and only if $\bar A = a\II - \vec{v}\cdot\vec\lambda \geq 0$.

The above reasoning proves the following particular case of Proposition \ref{prismo}.

\begin{prop}\label{landofooo} If a qudit measurement $\bA$, given by $A_i = a_i\II + \vec{u}_i\cdot\vec\lambda$, is simulable via unbiased 2-outcome POVMs, then the antipodal operators $\bar{A}_i = a_i\II - \vec{u}_i\cdot\vec\lambda$ are positive semidefinite, $\bar \bA$ is a valid POVM, and $\{\bA, \bar \bA\}$ is jointly measurable.
\end{prop}

Proposition \ref{landofooo} is an example of the power of Proposition \ref{prismo} that has a clear geometrical interpretation.
For the particular case of qubit measurements, in Section \ref{candykingdom} we are able to show its converse (see Theorem \ref{theenchiridion}).

\section{Limiting the mathematical structure of the simulators}\label{candykingdom}

We now investigate the case where the simulating set $\cB$ is constrained to have only projective POVMs.
This automatically limits the number of outcomes to be at most equal to the dimension of the system ($k\leq d$). 
In this case, a $\cB$-simulable measurement is said to be \textit{projective-simulable}~\cite{oszmaniec16}. 
Apart from their fundamental importance, projective measurements are often much easier to implement, as they do not require any ancilla system.

We recall now the following well-known result (Lemma 2.3 of Ref.~\cite{davies76}), due to the fact that extremal dichotomic POVMs are projective.
\begin{lemma}\label{rootbeerguy}
For any dimension $d$, any 2-outcome POVM is projective-simulable. 
If $d=2$, then 2-outcome simulability and projective simulability are equivalent.
\end{lemma}

In Ref.~\cite{oszmaniec16} characterisations of the set of projective-simulable POVMs were presented in dimension $d=2,3$ that allow one to efficiently decide by SDP whether a fixed POVM is simulable or not.
It was also shown that the tetrahedral qubit POVM $\bM^{\text{tetra}}$ (see section \ref{firekingdom}) is the most robust in terms of projective simulability~\cite{oszmaniec16, hirsch17}.

\subsection{Projective simulability and joint measurability}

In dimension $d=2$, we see that the unbiased 2-outcome measurements in Eq. (\ref{lumpspaceprincess}) are exactly the projective POVMs and their depolarised versions.
In this particular case, where 2-outcome and projective-simulability coincide (Lemma \ref{rootbeerguy}), we can prove the converse of Proposition \ref{landofooo}, which completely characterises the projective-simulable qubit POVMs.

\begin{theo}\label{theenchiridion}
A qubit POVM is projective-simulable if and only if the pair $\{\bA, \bar{\bA}\}$ is jointly measurable, where $\bar{\bA}$ is the antipodal measurement of $\bA$.
\end{theo}

\begin{proof} The only if part is a particular case of Proposition \ref{landofooo}, so we need only to show the if part.

Assume that $\bA$ and $\bar{\bA}$ are jointly measurable and $\bM$ is a joint measurement for the pair, with $M_{ab} = m_{ab}\II + \vec{w}_{ab}\cdot \vec{\sigma}$.
Consider now $N_{ab} = (M_{ab} + \bar{M}_{ba})/2$, where $\bar{M}_{ba}$ represents the antipodal operator of $M_{ba}$.
We have that $\bN$ is also a joint POVM for the pair, since
\begin{subequations}
\beq
\sum_b{N_{ab}} = \frac 12\left(\sum_b{M_{ab}} + \sum_b{\bar M_{ba}}\right) = A_a
\eeq
\beq
\sum_a{N_{ab}} = \frac 12\left(\sum_a{M_{ab}} + \sum_a{\bar M_{ba}}\right) = \bar A_b,
\eeq
\end{subequations}
with the feature that symmetric effects sum up to a multiple of the identity, 
\beq\label{ladyrainicorn}
N_{ab} + N_{ba} = (m_{ab} + m_{ba})\II.
\eeq
Thus
\beq
\bA = \sum_{a\leq b}(m_{ab}+m_{ba})\bB^{(ab)},
\eeq
where the 
POVMs $\bB^{(ab)}$ are defined by
\beq
B^{ab}_s = 
\begin{cases}
N_{ab}/(m_{ab}+m_{ba}),\ \ &\text{if}\ \ s=a \\
N_{ba}/(m_{ab}+m_{ba}),\ \ &\text{if}\ \ s=b \\
0,\ \ &\text{otherwise}  \\
\end{cases},
\eeq
and therefore can be interpreted as 2-outcome measurements embedded in the space of $n$-outcome POVMs.
The normalization of $\bN$ implies that $\sum_{a, b}m_{ab}=1$, which ensures that the decomposition is convex. 
Finally, since every 2-outcome measurement is projective-simulable (Lemma \ref{rootbeerguy}), we conclude that $\bA$ is projective-simulable.

\end{proof}

As a corollary, we see that the pair of antipodal tetrahedrons is the most robust pair of antipodal qubit measurements regarding joint measurability, since $\bM^{\text{tetra}}$ is the most robust qubit POVM regarding projective simulability~\cite{oszmaniec16, hirsch17} (Figure \ref{ohmyglob}).

\begin{coro}
$\{\bM^\text{tetra}, \bar{\bM}^\text{tetra}\}$ is the most robust pair of antipodal qubit measurements regarding joint measurability.
\end{coro}

\begin{figure}[h!]
\begin{center}
	\includegraphics[width=350pt]{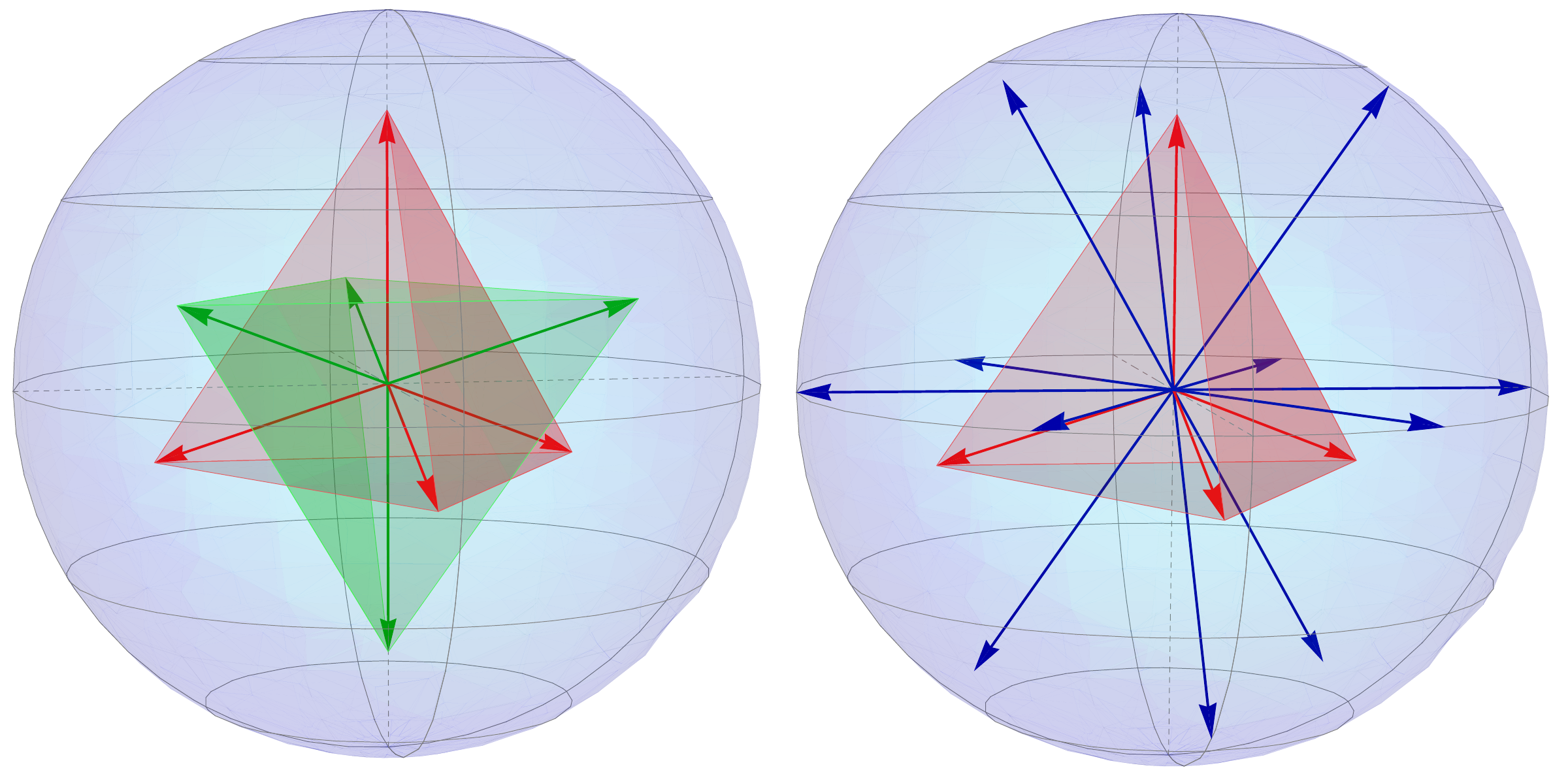}
	\caption{\small{Two antipodal tetrahedral measurements and their optimal projective simulators, which define a joint measurement for the tetrahedral pair. Both joint measurability with the antipodal and projective simulability are achieved at the same critical visibility $t=\sqrt{2/3}$.}}\label{ohmyglob}
\end{center}
\end{figure}

Another consequence of Theorem \ref{theenchiridion} is given by the close connection between joint measurability and EPR steering, namely that a set of POVMs is jointly measurable if and only if it cannot demonstrate steering when applied to any quantum state \cite{quintino14,uola14}.
Hence we see that projective simulability is also connected to EPR steering.

\begin{coro}
A qubit measurement $\bA$ is projective-simulable if and only if the pair $\{\bA, \bar \bA\}$ cannot demonstrate steering when applied to any quantum state of local dimension 2. 
\end{coro}

\section{Quantum measurement simulability as a resource theory}

The approach to measurement simulability we use here is close to a resource theory in many aspects.
A resource theory is a formal framework to study a given property of a class of objects, which plays the role of \textit{resource}.
The framework is defined by a subset of operations called \textit{free operations}, that has the key feature of not being able to generate the resource.
This means that when a free operation is applied to a \textit{free object}, \ie, an object without the property of interest, the resulting object is also free.
This approach was succesfully used to investigate properties such as entanglement~\cite{vedral97, plenio07}, thermal equilibrium~\cite{brandao13}, asymmetry~\cite{ahmadi13}, reference frames~\cite{gour08}, and nonlocality~\cite{devicente14}.

In our case, for every type of simulators $\cB$ we can define a resource theory where the resource is the non-$\cB$-simulability.
In the case of $J$-POVM simulability (Section \ref{lumpyspacekingdom}), the objects are sets of quantum measurements, the free operations are classical processing, and sets of $J$ measurements are free objects, implying that every simulable set is also free.
Analogously, in the case of $k$-outcome and projective simulability (Sections \ref{firekingdom} and \ref{candykingdom}), the objects are single measurements, and the free operations and objects are again classical processing and simulable measurements, respectively.

To formalise these notions, we prove now the invariance of the set of simulable POVMs by classical processing.
We show that the simulability relation is transitive, namely that if a set of measurements is $\cB$-simulable, then any classical manipulation of it is $\cB$-simulable as well.
This encompasses $J$-POVM simulability of sets of POVMs as a particular case, as well as $k$-outcome and projective simulability of single POVMs.

\begin{prop}\label{turtleprincess}
Let $\cB$ be a subset of measurements.
If a set of measurements $\cA = \{\bA^{(l)}\}$ is $\cB$-simulable, then any set $\tilde\cA$ obtained by classically processing $\cA$ is $\cB$-simulable as well.
\end{prop}
\begin{proof}
Suppose $\tilde\cA$ contains POVMs $\tilde \bA^{(l)}$, constructed by pre- and post-processing the elements of $\cA$,
\beq\label{breakfastprincess}
\tilde A^{(k)}_{a_k} = \sum_{l}p'(l|k)\sum_{a_l} q'(a_k|k,l,a_l)A^{(l)}_{a_l},
\eeq
for all outcomes $a_k$ and for some probability distributions $p'(\cdot|k),q'(\cdot|k,l,a_l)$, where $k$ runs over the number of elements of $\cA'$, $l$ runs over the number of elements of $\cA$, and $a_l$ runs over the outcomes of $\bA^{(l)}$.
Since we can simulate $\cA$ using $\cB$, there are probability distributions $p(\cdot|l),q(\cdot|l,j,b_j)$, where $j$ labels a POVM $\bB^{(j)}\in\cB$ and $b_j$ its outcomes, satisfying Eq. (\ref{princessbubblegum}).
Thus we can substitute it in the above equation, yielding
\beq
\tilde A^{(k)}_{a_k} = \sum_j\tilde p(j|k)\sum_{b_j}\tilde q(a_k|k,j,b_j)B^{(j)}_{b_j},
\eeq
where
\begin{eqnarray}
\tilde p(\cdot|k) &:=& \sum_{l}p'(l|k)p(\cdot|l) \\
\tilde q(\cdot|k,j,b_j) &:=& \sum_{a_l} q'(\cdot|k,l,a_l) q(a_l|l,j,b_j)
\end{eqnarray}
define pre- and post-processings that simulate $\tilde\cA$ with $\cB$.
\end{proof}

A secondary but still important element of a resource theory is a way of quantifying the resource.
A quantifier function must be monotonic with respect to the free operations, meaning that by performing a free operation one should not be able to increase the measured quantity of resource of the initial object.

Usually the same theory allows many different quantifiers.
We finish this section showing that, for measurement simulability, the white noise robustness of a set of measurements (see Eq.(\ref{bananaguard})) is a suitable measure of non-simulability.
\begin{prop}
The white noise robustness of a set of POVMs regarding $\cB$ simulability is monotonic with respect to classical processings.
\end{prop}
\begin{proof}
Suppose $\tilde\cA$ is obtained by classical processing $\cA$.
Following Eq. (\ref{breakfastprincess}), we have
\beq
\Phi_t(\tilde A^{(k)}_{a_k}) = \sum_{l}p'(l|k)\sum_{a_l} q'(a_k|k,l,a_l)\Phi_t(A^{(l)}_{a_l}).
\eeq
This implies that at the critical visibility $t_\cB^\cA$ that makes $\cA$ $\cB$-simulable we can write each effect $\Phi_t(A^{(l)}_{a_l})$ as an appropriate combination of effects of the simulators, and then substitute in the previous equation to find that $\tilde\cA$ is also $\cB$-simulable.
Therefore, $t_\cB^{\tilde\cA} \geq t_\cB^\cA$.
\end{proof}

\section{Discussion}

We presented an operational framework for simulating quantum measurements that comprehends well-known scenarios in the field as particular cases, and identified different connections between them.
This allowed us to describe $k$-outcome simulability 
and projective simulability for qubit POVMs in terms of joint measurability, which appears as a common denominator in this context.


With Theorem \ref{theenchiridion}, we showed an equivalence between projective simulability and joint measurability for qubit measurements.
It remains as an open problem whether there is a similar characterisation for projective simulability in higher dimensions.
One would need to find a proper generalisation for the antipodal POVM that is well-defined for any POVM, in any dimension.
The antipodal operator itself is related to the universal-NOT gate \cite{buzek99}, but a straightforward generalisation is not related to projective simulability. 
In dimension $d=2$, by using the strong connection between joint measurability and EPR steerability~\cite{quintino14,uola14}, as a consequence we also have a relation between projective simulability and EPR steering.

Finally, we also discuss how our approach can be interpreted in the context of resource theories.
Exploring this connection in more detail seems to be a promising direction for future work.

\section*{ACKNOWLEDGMENTS}

The authors thank Marcus Huber, Joseph Bowles, Daniel Cavalcanti, Marco T\'ulio Quintino, and Micha\l{} Oszmaniec for interesting discussions, and Rafael Rabelo and Alexia Salavrakos for helping to improve this manuscript.
This work was supported by the Brazilian agencies CAPES, CNPq, and FAEPEX, Spanish MINECO (QIBEQI FIS2016-80773-P and Severo Ochoa SEV-2015-0522), Fundaci\`o Cellex, Generalitat de Catalunya (SGR875 and CERCA Program), ERC CoG QITBOX, AXA Chair in Quantum Information Science, and the Austrian Science Fund (FWF) through the START project Y879-N27.


\begin{thebibliography}{10}
\bibitem{stupidfootnote}
{The set $\mathcal{B}$ needs not to be countable, as it happens, \eg, for
  projective simulable measurements. However, given a measurement $\mathbf{A}$
  and a set of simulators $\mathcal{B}$, Caratheodory's Theorem guarantees that
  a finite subset of $\mathcal{B}$ is enough. This justifies our use of sums
  instead of integrals.}
\end{thebibliography}

\end{document}